\documentclass[11pt]{article}
\usepackage{geometry}
\usepackage[utf8x]{inputenc} 

\usepackage{graphicx}
\usepackage{colordvi}
\usepackage{color}
\usepackage{paralist}
\usepackage{bm}
\usepackage{hyperref}
\usepackage{csquotes}
\usepackage{mathpazo}
\usepackage{epsfig}
\usepackage{amsmath,amssymb}
\usepackage{amsfonts}
\usepackage{amssymb}
\usepackage{amstext}
\usepackage{amsmath}
\usepackage{xspace}
\usepackage{mathtools}
\usepackage{algorithm}
\usepackage[noend]{algpseudocode}

\newcommand{\abs}[1]{\mbox{$\left|#1\right|$}}

\newcommand{\set}[1]{\mbox{$\left\{#1\right\}$}}

\newcommand{\remove}[1]{}
\newtheorem{fact}{Fact}[section]

\newtheorem{definition}[fact]{Definition}
\newtheorem{claim}{Claim}
\newtheorem{lemma}{Lemma}
\newtheorem{theorem}{Theorem}
\def\Box{\rule{2mm}{2mm}}

\newenvironment{proof}{\noindent {\it Proof.}}{\Box \vskip \belowdisplayskip}

\newcommand{\mcR}{\mathcal{R}} 
\newcommand{\mcS}{\mathcal{S}} 
\newcommand{\mcT}{\mathcal{T}} 
\newcommand{\mcA}{\mathcal{A}} 
 
\newcommand{\mcC}{\mathcal{C}} 

\newcommand{\junk}[1]{}
\newenvironment{lp}[2]{\[\begin{array}{rcll}
                        \mbox{#1} & & #2 \\ 
                        \mbox{subject to}}{\end{array}\]}
\newcommand{\cnstr}[4]{\\ #1 & #2 & #3 & #4}

\def\UCC{\mbox{\tt UCC}}

\begin{document}

\title{Dual Half-integrality for Uncrossable Cut Cover and its Application to Maximum Half-Integral Flow}

\author{Naveen Garg, Nikhil Kumar \\ Indian Institute of Technology Delhi, India}

\maketitle

\begin{abstract}
Given an edge weighted graph and a forest $F$, the {\em 2-edge connectivity augmentation problem} is to pick a minimum weighted set of edges, $E'$, such that every connected component of $E'\cup F$ is 2-edge connected. Williamson et al. gave a 2-approximation algorithm (WGMV) for this problem using the primal-dual schema. We show that when edge weights are integral, the WGMV procedure can be modified to obtain a half-integral dual. 
The 2-edge connectivity augmentation problem has an interesting connection to routing flow in graphs where the union of supply and demand is planar. The half-integrality of the dual leads to a tight 2-approximate max-half-integral-flow min-multicut theorem. 
\end{abstract}

\section{Introduction}
Let $G=(V,E)$ be an undirected graph with integer edge costs $c:E\rightarrow\mathbb{Z}^+$ and let $f:2^V\rightarrow\mathbb{Z}^+$ be a requirement function on sets of vertices. We wish to find a set of edges, $E'$ of minimum total cost such that for every set $S$ the number of edges in $E'$ across $S$ is at least the requirement of $S$, ie. $f(S)$. This problem captures many scenarios in network design and has been the subject of much investigation. The Steiner forest problem, minimum weight maximum matching and other problems can be modeled by requirement functions which are {\em proper} and 0-1 (see Definition \ref{def:proper}) and for such functions Agrawal, Klein, Ravi~\cite{AKR} and Goemans, Williamson~\cite{GW} gave a primal-dual algorithm that is a 2-approximation. The key idea of primal-dual algorithms is to use complementary slackness to guide the construction of the dual and primal solutions which are within a factor 2 of each other. 

To use this approach for the Steiner network design problem where the requirements of sets are not just 0-1, Williamson et al.~\cite{WGMV} extend the primal dual algorithm of GW to the setting of 0-1 {\em uncrossable} requirement functions (see Definition \ref{def: uncrossable}); we call this the WGMV algorithm. The idea was to augment the connectivity of the solution in rounds with each round augmenting the requirements of unsatisfied sets by 1. The WGMV algorithm for uncrossable functions also builds a dual solution and while the primal solution constructed is integral, nothing is known of the integrality of the dual solution. In particular while for proper functions it is possible to argue that the dual solution constructed by the GW procedure is half-integral the same is not true for the WGMV procedure for uncrossable functions as is illustrated by the example in Section \ref{WGMV:counter}. 


For weakly supermodular requirement functions (see Definition \ref{def:weakly-sup}) Jain~\cite{Jain95} gave a 2 approximation algorithm based on iterative rounding. Although this algorithm does not build a dual solution, the iterative rounding technique saw a lot of interesting applications and quickly became an integral part of tool-kit of approximation algorithms. This together with the fact that the dual solution constructed by the WGMV procedure seems useful only for certifying the approximation guarantee of the procedure, implied that there were no further results on the nature and properties of the dual solution. 


In~\cite{KGS} the authors show that the problem of finding maximum multiflow when the union of the supply and demand edges forms a planar graph can be reduced to the problem of finding a large dual solution for a suitable cut-covering problem with uncrossable requirement function. In addition, a primal solution would correspond to a multicut and a half-integral dual solution would correspond to a half-integral multiflow. Therefore, a primal solution which is within twice a half-integral dual solution would imply a 2-approximate max-half-integral-multiflow min-multicut theorem for such graph classes. In \cite{KGS} the authors also show instances where max-half-integral-multiflow min-multicut gap can be arbitrarily close to 2, implying that our result is best possible. 

In this paper we show that a suitable modification to the WGMV procedure does indeed lead to a half-integral dual solution of value at least half the primal solution. 
\begin{theorem} \label{thm : half integer WGMV}
Let $G=(V,E)$ be an undirected graph with edge costs $c:E\rightarrow\mathbb{Z}^+$ and a uncrossable requirement function $f:2^V\rightarrow\set{0,1}$. One can find a subset of edges $F$ and an assignment, $y$, of non-negative half-integral dual variables to sets such that for all edges $e\in E$, $\sum_{S:e\in\delta(S)} y_S\le c_e$ and $\sum_{e\in F} c_e\le 2\sum_S f(s)y_S$.
\end{theorem}
To achieve this we need to build an alternate stronger analysis of the 2-approximation of the WGMV algorithm and these are the main results of this paper. In Section~\ref{sec:proper} we argue that the Goemans-Williamson algorithm for proper functions leads to half-integral duals. To prove the above, we come up with a notion of parity of a node with respect to the current dual solution. The crux of our argument is to show that all nodes in an active set have the same parity. We then employ the idea of ensuring that all nodes in an active set have the same parity to modify the WGMV procedure in Section~\ref{sec:WGMVmodified}. However our procedure for ensuring uniform parity entails reducing some edge costs by 1/2. Since this decrease in edge costs also needs to be bounded by the dual solution we need a stronger guarantee on the total degree of the active sets in each iteration of the WGMV procedure. We develop this alternate analysis in Section~\ref{sec:WGMVanalysis}. Finally, Section~\ref{sec:Seymour} shows how maximum flow in Seymour graphs corresponds to building the dual solution for a suitable uncrossable cut cover problem and lets us claim the following result which is also best possible.
\begin{theorem} \label{half integral maxflow mincut}
Let $G+H$ be planar. There exists a feasible half-integral flow of value $F$ and a multicut of value $C$ such that $C \leq 2F$. Further, such a flow and cut can be computed in polynomial time.
\end{theorem}

\section{Preliminaries}
Given a graph $G=(V,E)$ with edge costs $c:E\rightarrow\mathbb{R^+}$ and a 0-1 requirement function $f:2^V\rightarrow\{0,1\}$ we are interested in picking a subset of edges $E'$ of minimum total cost such that every set with requirement 1 has at least one edge of $E'$ across it. In other words, for all $S\subseteq V$, $|\delta_{E'}(S)|\ge f(S)$, where $|\delta_{E'}(S)|$ is the number of edges in $E'$ which have exactly one endpoint in $S$.

\begin{definition} \label{def:proper}
A function $f:2^{V}\rightarrow \{0,1\}$ is called \textbf{proper} if $f(V)=0,f(S)=f(V-S)$ for all $S \subseteq V$ and for any disjoint $A,B \subseteq V$, $f(A \cup B)\leq \max \{f(A),f(B) \}$.
\end{definition}

\begin{definition}  \label{def:weakly-sup}
A function $f:2^{V}\rightarrow \{0,1\}$ is called \textbf{weakly supermodular} if $f(V)=f(\phi)=0$ and for any $A,B \subseteq V$, $f(A)+f(B) \leq \max \{f(A \cap B)+f(A \cup B),f(A \setminus B) + f(B \setminus A)\}$.
\end{definition}

\begin{definition}  \label{def: uncrossable}
A function $f:2^{V}\rightarrow \{0,1\}$ is called \textbf{uncrossable} if $f(V)=f(\phi)=0$ and for any $A,B \subseteq V$, if $f(A)=f(B)=1$, then either $f(A \cap B)=f(A \cup B)=1$ or $f(A \setminus B)=f(B \setminus A)=1$.
\end{definition}

It is easy to argue that every proper function is also weakly supermodular and every weakly supermodular function is also uncrossable. In this paper we will only be interested in uncrossable requirement functions and shall refer to the problem in this setting as the {\em uncrossable cut cover problem} (\UCC). The following integer program for $\UCC$ is well known.
\begin{lp}{minimize}{\sum_{e \in E} c_e x_e}
\cnstr{\sum_{e \in \delta(S) } x_e}{\ge}{f(S)}{ S \subseteq V}
\cnstr{x_e}{\in}{\{0,1\}}{e \in E}
\end{lp}
We can relax the integrality constraint on $x_{e}$ to $0 \leq x_{e} \leq 1$ to get a linear programming relaxation of the above. The dual program of the relaxation can be given as: 
\begin{lp}{maximize}{\sum_{S \subseteq V} f(S)y(S)}
\cnstr{\sum_{S: e \in \delta(S) } y_S}{\le}{c_{e}}{e \in E}
\cnstr{y_S}{\ge}{0}{S \subseteq V}
\end{lp}
Williamson et al. \cite{WGMV} gave a primal-dual 2-approximation algorithm for the above integer program for uncrossable $f$.

\section{Half-integrality of the GW-dual for proper functions}
\label{sec:proper}
We first argue that the Goemans-Williamson (GW) algorithm - for the case when requirement functions are proper and edge costs are integral - constructs a half-integral dual whose value is at least half the primal integral solution.

The GW algorithm proceeds by raising dual variables corresponding to sets of vertices and picking edges which are tight into the current solution. An edge $e$ is tight when the sum of dual variables of sets containing exactly one end-point of $e$ equals $c(e)$. The algorithm raises dual of all minimal sets $S$ such that $f(S)=1$ but no edge going across $S$ has been picked in the current solution. We imagine growing the duals in a continuous manner and define a notion of time: $t=0$ at start of the algorithm and $y_S$ increases by $\delta$ during $[t,t+\delta]$ if $S$ is a minimally unsatisfied set at every point of time in $[t,t+\delta]$. If $f$ is proper, these minimal sets correspond exactly to the connected components formed by the set of tight edges. Let $C$ be a connected component at time $t$. If $f(C)=1$ then $C$ is active while $C$ is inactive if $f(C)=0$. In each iteration, the GW procedure raises dual variables of all active sets simultaneously till an edge goes tight. At this point the connected components are recomputed and the algorithm continues with the next iteration unless all sets are inactive. Let $F$ be the set of tight edges picked after the first phase. In a second phase, called the {\em reverse delete}, the GW algorithm considers the edges of $F$ in the reverse order in which they were added to $F$. If the removal of an edge from $F$ does not violate the requirement function of any set then the edge is removed.

We shall only be concerned with the first phase of the GW algorithm since it is in this phase that the dual variables, $y:2^V\rightarrow\mathbb{R}_{\ge 0}$ are set. Let $\mcS=\set{S:y_S>0}$ and note that this family of sets is laminar. For $v\in S, S\in\mcS$, we define the parity of $v$ with respect to $S$ as $\pi_v(S)=\left\{\sum_{T:v\in T\subseteq S}y_T\right\}$, where $\{x\}$ denotes the fractional part of $x$. 
If $S$ is active at time $t$ then there exists a vertex $v\in S$ which for all times in $[0,t]$ was in an active component; we call such a vertex an {\em active vertex} of set $S$. 

We now argue that the GW procedure ensures that for all $S\in\mcS$, for all $ u,v\in S$, $\pi_u(S)=\pi_v(S)$. We call this quantity the parity of set $S$, $\pi(S)$, and show that $\pi(S)\in\set{0,1/2}$. Let $S$ be formed by the merging of sets $S_1,S_2$ at time $t$. We induct on the iterations of the GW procedure and assume that all vertices in $S_1$ (respectively $S_2$) have the same parity with respect to $S_1$ (respectively $S_2$). If $S_1$ is active at time $t$ then $\pi(S_1)=\pi_v(S_1)=\{t\}$ where $v$ is an active vertex of set $S_1$. Similarly if $S_2$ is active at time $t$ then $\pi(S_2)=\{t\}$. Thus if both $S_1,S_2$ are active at time $t$ then $\pi(S_1)=\pi(S_2)$ and hence all vertices of $S$ have the same parity with respect to $S$. Let $e=(u,v), u\in S_1, v\in S_2$ be the edge which gets tight when $S_1,S_2$ merge at time $t$. Since $l(e)$ is integral and $\sum_{T:u\in T\subseteq S_1} y_T+\sum_{T:v\in T\subseteq S_2} y_T=l(e)$, we have that $\pi(S_1)=\pi(S_2)\in\set{0,1/2}$.

Suppose only $S_1$ is active at time $t$. By our induction hypothesis $\pi(S_2)\in\set{0,1/2}$. Once again, since $l(e)$ is integral and $\sum_{T:u\in T\subseteq S_1} y_T+\sum_{T:v\in T\subseteq S_2} y_T=l(e)$, we have that $\pi(S_1)=\pi(S_2)$ which implies that all vertices of $S$ have the same parity with respect to $S$.

Since $\pi(S),\pi(S_1)\in\set{0,1/2}$, it must be the case that $\set{y_S}\in\set{0,1/2}$. Since this is true for all sets $S\in\mcS$ this implies that the duals constructed by the GW procedure are half-integral.

\section{The WGMV algorithm}
We now give a brief description of the algorithm in \cite{WGMV}. Given an undirected graph $G=(V,E)$ with edge costs $c_{e} \geq 0$ and a uncrossable function $f$ we wish to find a  set of edges $F' \subseteq E$ such that for any $S \subseteq V, |F' \cap \delta(S)| \geq f(S)$. A set $S$ is said to be unsatisfied if $f(S)=1$ but no edge crosses $S$ in the current solution. 

The algorithm works in iterations. At the beginning of every iteration the algorithm computes a collection of minimally unsatisfied sets. Williamson et al.\cite{WGMV} show that minimally unsatisfied sets are disjoint and can be found in polynomial time (follows easily from uncrossability). Raise the dual variables corresponding to all minimally unsatisfied sets simultaneously until some edge is tight (the total dual across it equals its cost). All edges that go tight are added to a set $T$. The edges of $T$ are considered in an arbitrary order and $e\in T$ is added to $F$ if it crosses a minimally unsatisfied set. Note that whenever an edge is added to $F$ the collection of minimally unsatisfied sets is recomputed. The {\em growth phase} of the WGMV procedure stops when all sets are satisfied; let $F$ be the set of edges picked in this phase. 

The edges of $F$ are considered in the reverse order in which they were picked. An edge $e\in F$ is dropped from the solution if its removal keeps the current solution feasible. 

At the end of the procedure, we have a set of edges $F$ and a feasible dual solution $y_S$ such that $\sum_{e \in F} c_{e} x_{e} \leq 2 \sum_{S}f(S)y_S$. By weak duality, $\sum_{e \in F} c_{e} x_{e} \geq  \sum_{S}f(S)y_S$ and this shows that the cost of solution picked by the algorithm is at most twice the optimal. 
\begin{algorithm}
\caption{Primal-Dual Algorithm for uncrossable functions}\label{alg:euclid}
\begin{algorithmic}[1]
\Procedure{WGMV} { $G=(V,E)$ with cost $c_{e}$, uncrossable function $f$}
\State $y\leftarrow 0$, $F \leftarrow \phi$
\While {$\exists S \subseteq V$ such that $S$ is not satisfied}
\State Compute $\mcC$, the collection of minimally unsatisfied sets with respect to $F$.
\State Increase $y_{C}$ for all $C \in \mathcal{C}$ simultaneously until some edge $e \in \delta(C), C \in \mathcal{C}$ is tight ($c_e=\sum_{S:e\in \delta(S)}y_S$)
\State  Add all tight edges to $T$
\ForAll{$e\in T$}
\If{$\exists C\in\mcC, e\in\delta(C)$} 
\State $F\leftarrow F\cup\set{e}$; Recompute $\mcC$
\EndIf
\EndFor
\EndWhile
\ForAll{$e\in F$}
\State // Edges of $F$ are considered in the reverse order in which they were added to $F$
\If{$F\setminus\set{e}$ is feasible}
\State $F\leftarrow F\setminus\set{e}$
\EndIf
\EndFor
\State \textbf{return} $F$
\EndProcedure
\end{algorithmic}
\end{algorithm}

\subsection{Duals constructed by WGMV are not half-integral} \label{WGMV:counter}
In the example in Figure~\ref{fig:counter}, the red edges are not edges of the graph $G$. For a set $S\subseteq V$,$f(S)=1$ iff there is exactly one red edge with exactly one end point in $S$. Thus this problem corresponds to picking edges so as to augment the red tree into a 2-edge connected graph. It is known that $f$ is uncrossable. In each iteration the WGMV procedure raises dual variables corresponding to all minimally unsatisfied sets. The edge $(c,d)$ gets tight in the first iteration. At the end of the first iteration $y_{\set{e}}=1/2$ and so in the second iteration $y_{\set{e}}$ increases to 3/4 and $y_{\set{b,c,d}}$ to 1/4 before edge $(b,e)$ goes tight. 
\begin{figure}[h]
\centering
\includegraphics[scale=0.5]{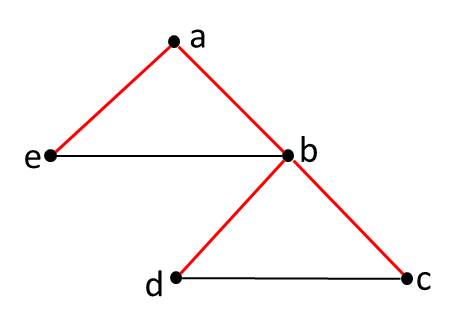}
\caption{Example showing that the duals constructed by the WGMV procedure are not half-integral\label{fig:counter}}
\end{figure}

\section{A stronger analysis of the WGMV algorithm}
\label{sec:WGMVanalysis}
To analyse the algorithm, Willimason et al.\cite{WGMV} argue that in each iteration the total contribution of the dual variables to the primal solution is at most twice the increase in the value of the dual solution. This then, added over all iterations, implies that $\sum_{e\in F} c_e x_e\le 2\sum_S f(S)y_S$. If in an iteration the dual values of all active sets increases by $\delta$ then the contribution of the dual variables to the primal solution equals $\delta$ times the total degree of the active sets in $F$. On the other hand the increase in the value of the dual solution is $\delta$ times the number of active sets and hence Williamson et al. argue that in each iteration the average degree of the active sets in $F$ is at most 2. 

Let $\mcS$ be the collection of minimally unsatisfied sets identified during a run of the algorithm. Note thst we do not claim that $y_S >0$ for $S\in\mcS$. The uncrossability of $f$ implies that $\mcS$ is a laminar family.  Add $V$, the set of all vertices, to $\mcS$ and construct a tree, $\mcT=(X,Y)$, which has vertex set $X=\set{v_S|S\in\mcS}$. $v_A$ is the parent of $v_B$ iff $A$ is the minimal set in $\mcS$ containing $B$. 

Each set $S\in\mcS$ is labelled with the number of the iteration in which $S$ became satisfied; let $l:\mcS\rightarrow[T]$ be this function. Let $\mcS^i$ be the sets with label at least $i$; these are the minimally unsatisfied sets encountered in iterations $i$ or later. Similarly, each edge $e\in F$ is labeled with the number of the iteration in which it became tight. We overload notation and let $l:F\rightarrow[T]$ also denote this function. Let $F^i\subseteq F$ be edges with label at least $i$.  We note a few properties of these labels.
\begin{enumerate}
    \item if $B\subset A$ then $l(B)\le l(A)$.
    \item if $e\in\delta_F(S)$ then $l(e)\ge l(S)$
\end{enumerate}

Let $v_{B_1},v_{B_2},\ldots v_{B_p}$ be the children of node $v_A$ in $\mcT$ (see Figure~\ref{fig:notation}). We number sets so that $l(B_1)\ge l(B_2)\ge \cdots \ge l(B_p)$. Let $p_i\in[p]$ be the largest index such that $l(B_{p_i})\ge i$. Hence all sets $B_j, j\in[p_i]$ are in $\mcS^i$.
Let $X^i_A=A\setminus\cup_{j\in [p_i]} B_j$ and $H^i_A$ be a graph whose nodes correspond to sets $X^i_A, B_1,\ldots, B_{p_i}$ and edges correspond to the edges between these sets in $F$. Since sets $B_j, j\in[p_i]$ have label at least $i$,  edges in $H^i_A$ have label at least $i$ and hence they are in $F^i$.  

\begin{claim}
$H^i_A$ is a forest.
\end{claim}

\begin{proof}
For contradiction assume $H^i_A$ has a cycle and consider the edge of the cycle, say $e=(u,v)$, which was added last to $F$. We consider two cases.

$u \in B_r$ and $v \in B_s$, $r,s\in[p_i]$. When $e$ was picked, both $B_r,B_s$ had another edge in $F$ across them and were therefore satisfied. Recall that $\mcS$ is the collection of all the minimally unsatisfied sets encountered during the growth phase of the algorithm. Picking $e$ did not lead to any unsatisfied set in $\mcS$ getting satisfied and this is a contradiction. 

$u \in B_r$ and $v \in X^i_A$, $r\in[p_i]$. No subset of vertices in $X^i_A$ is unsatisfied in the $i^{\rm th}$ (or any subsequent) iteration. When $e$ was picked, $B_r$ had another edge in $F$ across it and was therefore satisfied. Once again picking $e$ did not lead to any unsatisfied set in $\mcS$ getting satisfied and this is a contradiction.
\end{proof}

Since $H^i_A$ is a forest on $p_i+1$ vertices it contains at most $p_i$ edges. 

\begin{definition}
A set $A\in\mcS$ is {\em critical} in iteration $i$ if $H^i_A$ is a tree of which the node corresponding to $X^i_A$, is a leaf.
\end{definition}
For a set $A\in\mcS^i$, let $\alpha^i(A)= \delta_{F}(A)\setminus\cup_{S\subset A, S\in\mcS^i} \delta_{F}(S)$. Thus $\alpha^i(A)$ is the set of edges of $F$ which have one endpoint in the set $A\setminus\cup_{S\subset A, S\in\mcS^i} S$ and the other endpoint in $V\setminus A$. Equivalently $\alpha^i(A)$ is the subset of edges in $\delta_F(A)$ which are incident on vertices in $X^i_A$. We note the following important property of $\alpha^i(A)$.
\begin{claim}
Let $A\in\mcS^i$. The collection of sets $\set{\alpha^i(S)| S\in\mcS^i, S\subseteq A}$ forms a partition of the set $\delta_F(A)$.
\end{claim}

Let $\mcA^i$ be the collection of minimally unsatisfied sets whose dual is raised in iteration $i$ of the WGMV algorithm. These are the {\em active sets} in iteration $i$. Note that 
\begin{enumerate}
    \item $\mcA^i\subseteq \mcS^i$. 
    \item A set $S \in \mathcal{S}$ is contained in $\mcS^i$ if and only if there exits an $A\in\mcA^i$ such that $A\subseteq S$.  
    \item If $A\in\mcA^i$ then no subset of $A$ is in $\mcS^i$ which implies $\alpha^i(A)=\delta_F(A)$.
\end{enumerate}

\begin{lemma}\label{lem:equal}
$\sum_{S\in\mcS^i} \abs{\alpha^i(S)} \le 2\abs{\mcA^i}-2+\abs{\mcR^i}$ where $\mcR^i$ is the collection of critical sets in iteration $i$. 
\end{lemma} 
\begin{proof}
We show an argument built on redistributing tokens which help us prove the above lemma. We begin by assigning every node of tree $\mcT$ a number of tokens equal to two less than twice the number of its children in $\mcS^i$. Thus a node with 1 child in $\mcS^i$ gets no tokens. We also give every node that corresponds to a critical set in iteration $i$ an additional token. It is easy to see that the total number of tokens distributed initially is $2\abs{\mcA^i}-2+\abs{\mcR^i}$.

$v_A$ transfers one token to each edge in $H^i_A$ incident on $X^i_A$ and 2 tokens each to remaining edges in $H^i_A$. If $v_A$ has $p_i$ children in $\mcS^i$ and is critical in iteration $i$, it was assigned $2p_i-1$ tokens and these are sufficient to undertake the above assignment. If $v_A$ is not critical then it was assigned $2p_i-2$ tokens and again this is sufficient to complete the transfer of tokens to edges in $H^i_A$. 

For every edge in $e\in F^i$ there is a unique $A\in\mcS^i$ such that $e$ is in $H^i_A$. If $e$ has an endpoint in $X^i_A$ it is assigned 1 token by $v_A$. Note that this edge contributes 1 to the sum on the left. The remaining edges of $F^i$ are assigned 2 tokens each and this is also their contribution to the sum on the left. This establishes that the sum on the left equals the number of tokens assigned to edges which is at most the number of tokens assigned to nodes which in turn equals the quantity in the right.
\end{proof}

\begin{lemma}\label{lem:one}
If $A$ is critical in iteration $i$ then $\alpha^i(A)\neq\phi$.
\end{lemma}
\begin{proof}
Since $A$ is critical, $H^i_A$ is a tree and $X^i_A$ is a leaf node. Let $e$ be the unique edge in $H^i_A$ incident to $X^i_A$. Consider the step in the reverse delete phase when edge $e$ was considered and was retained in $F$ only because its deletion would have caused some set to become unsatisfied. Let $U\subseteq V$ be the minimal such set and note that $e$ is the only edge in $F'$ across $U$ at this step. 

\begin{claim}
$\forall j\in[p_i]$, $U\cap B_j =\phi$ or $U\cap B_j=B_j$.
\end{claim}
\begin{proof}
For a contradiction assume that for some $j\in[p_i]$, $\phi\neq U\cap B_j\subset B_j$. Since $f(B_j)=f(U)=1$ by uncrossability either $f(B_j\cap U)=1$ or $f(B_j\setminus U)=1$. In either case, during the growth phase we must have added an edge, say $g$, to $F$ between $B_j\setminus U$ and $B_j\cap U$ in an iteration before $B_j$ became a minimally unsatisfied set. Thus, in the reverse delete phase when we considered $e$, edge $g$ was in $F$ and hence $e$ was not the only edge across $U$. 
\end{proof}

\begin{figure}
\centering
\includegraphics[scale=0.5]{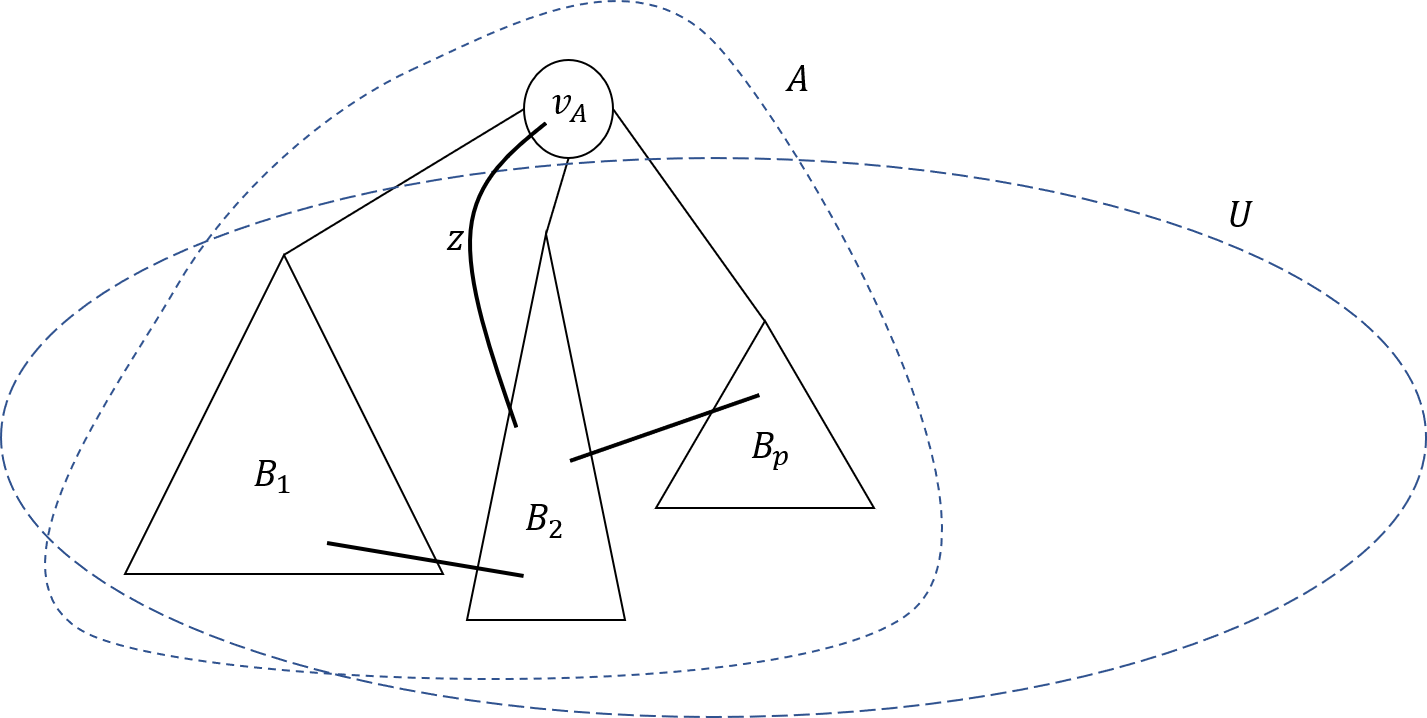}
\caption{\label{fig:notation} Illustrating the notation used. $A$ is a critical set. The thick edges are the edges in $H^i_A$.}
\end{figure}
If $U\cap A$ includes some sets $B_j, j\in[p_i]$ and not the others then the number of edges across the set $U$ will be more than 1. Thus either $\cup_{j\in p_i} B_j\subseteq A\cap U$ or $\cup_{j\in[p_i]} B_j\subseteq A\setminus U$. Since $f(A)=f(U)=1$ by uncrossability we have either $f(A\cap U)=1$ or $f(A\setminus U)=1$. If $\cup_{j\in[p_i]} B_i \subseteq A\cap U$ then $f(A\setminus U)\neq 1$ as that would imply a minimal unsatisfied set in $X^i_A$ which would be a contradiction. Similarly if $\cup_{j\in[p_i]} B_j\subseteq A\setminus U$ then $f(A\cap U)\neq 1$. Hence we need to consider only two cases
\begin{enumerate}
\item $\cup_{j\in[p_i]} B_j \subseteq A\cap U$ and $f(A\cap U)=f(A\cup U)=1$: $F$ should have an edge across the set $A\cup U$. Since the only edge across $U$ goes to $A\setminus U$, there should be an edge across $A$ that is incident to $X^i_A$.
\item $\cup_{j\in[p_i]} B_j \subseteq A\setminus U$ and $f(A\setminus U)=f(U\setminus A)=1$: $F$ should have an edge across the set $U\setminus A$. Since the only edge across $U$ goes from $A\cap U$ to $A\setminus U$, there should be an edge across $A$ that is incident to $A\cap U \subseteq X^i_A$.
\end{enumerate}
Hence in both cases we conclude that there is an edge across $A$ incident to $X^i_A$ which implies $\alpha^i(A)\neq\phi$. 
\end{proof}

\begin{lemma}
\label{lem:twice}
The total degree (in $F$) of sets in $\mcA^i$ is at most twice $\abs{\mcA^i}$.
\end{lemma} 
\begin{proof}
A set in $\mcA^i$ cannot be critical in iteration $i$. Further for $S\in\mcA^i$, $\abs{\alpha^i(S)}$ equals the degree of $S$ in $F$. By Lemma~\ref{lem:one} if $A$ is critical in iteration $i$ then $\alpha^i(A)\neq\phi$. Hence $\sum_{S\in\mcS^i} \abs{\alpha^i(S)} \ge \sum_{S\in\mcA^i} \abs{\delta_F(A)}+\abs{\mcR^i}$ where $\mcR^i$ is the collection of critical sets. Applying Lemma~\ref{lem:equal}, we obtain $\sum_{S\in\mcA^i} \abs{\delta_F(A)}\le 2\abs{\mcA^i}-2$ which proves the lemma.
\end{proof}
\remove{
$A$ became minimally unsatisfied only after sets $B_1,B_2,\ldots, B_p$ were satisfied by inclusion of some tight edges in $F$; let $e\in F$ be one such tight edge. Note that $e$ has both endpoints within the set $A$ since if one endpoint of $e$ was not in $A$ then $A$ would have been satisfied when $e$ was added to $F$ and would never become a minimally unsatisfied set.} 

\section{Modifying WGMV}
\label{sec:WGMVmodified}
We now modify the WGMV algorithm so that the duals obtained are half-integral while ensuring that the primal solution has cost at most twice the dual solution. In doing so we are guided by the fact that the GW algorithm constructed half-integral duals since the parity of all vertices in a set was identical. This property does not hold true for the WGMV algorithm as seen in the example in Figure~\ref{fig:counter}. 

As before, let $\mathcal{S}$ be the set of minimally unsatisfied sets during a run of the algorithm. Our modification to the WGMV algorithm involves reducing costs of some edges in $\delta(S), S\in\mcS$ by 1/2. Let $\delta'(S)\subseteq\delta(S)$ denote the subset of edges whose cost was reduced by 1/2 when considering $S$. We now define the {\em parity} of an edge $e$ with respect to a set $S\in\mcS, e\in\delta(S)$ as $$\pi_e(S)=\set{\sum_{e\in\delta(T),T\subseteq S} y_T +\frac{1}{2}\abs{\set{T\subseteq S|e\in\delta'(T)}}}$$ 
where as before \set{x} denotes the fractional part of $x$. Our modification to the WGMV procedure is:
\begin{quote}
Let $S$ be a set which becomes minimally unsatisfied at time $t$ and let $x\in S$ be an active vertex of set $S$. Then $\pi_x(S)=\set{t}$. For edge $e\in\delta(S)$, if $\pi_e(S)\neq\set{t}$ then decrease $c_e$ by 1/2 (note $e$ gets included in $\delta'(S)$).  
\end{quote}
We decrease the costs of edges in $\delta(S)$ in the above manner only when $S$ becomes minimally unsatisfied and need to argue that the total cost of edges in $F$ can still be bounded by twice the sum of the dual variables. Our modification allows us the following claim.
\begin{claim}
$\forall S\in\mcS$, $\forall e,f\in\delta(S)$, $\pi_e(S)=\pi_f(S)$
\end{claim}

When we increase dual variables of sets in $\mcA^i$ in iteration $i$, one or more edges go tight and these are added to a set $T$. Let $t^i$ be the time at which we stop growing dual variables of sets in $\mcA^i$. The edges of $T$ are considered in an arbitrary order and $e\in T$ is added to $F$ if it crosses a minimally unsatisfied set. Note that whenever an edge is added to $F$ the collection of minimally unsatisfied sets is recomputed. Let $\mcC$ be the collection of minimally unsatisfied sets after all edges in $T$ have been considered. For every $S\in\mcC$ and every edge $e\in\delta(S)$, if $\pi_e(S)\neq\set{t}$ then we reduce the cost of edge $e$ by 1/2. All edges that go tight after this step are included in $T$ and the process repeated until no edge gets added to $T$. The minimally unsatisfied sets at this stage are the active sets, $\mcA^{i+1}$ for iteration $i+1$. 

\begin{algorithm}
\caption{Modification to an iteration of the WGMV algorithm}
\label{alg:modified}
\begin{algorithmic}[1]
\State $\mcC$ is the collection of minimally unsatisfied sets with respect to $F$.
\State $T$ is the set of tight edges which have not yet been included in $F$.
\Repeat 
\ForAll{$e\in T$}
\If{$\exists C\in\mcC, e\in\delta(C)$} 
\State $F\leftarrow F\cup\set{e}$; compute $\mcC$
\EndIf
\EndFor
\State $T\leftarrow\phi$
\ForAll{$C\in\mcC$}
\ForAll{$e\in\delta(C)$}
\If{$\pi_e(C)\neq\set{t}$}
\State $c_e\leftarrow c_e-1/2$
\If{$e$ is tight}
\State $T\leftarrow T\cup\set{e}$
\EndIf
\EndIf
\EndFor
\EndFor
\Until {$T=\phi$}
\end{algorithmic}
\end{algorithm}
Let $\mcC^i$ be the collection of sets in $\mcS$ which properly contain a set in $\mcA^i$ and are subsets of some set in $\mcA^{i+1}$. Formally,  $\mcC^i=\set{S\in\mcS|\exists A\in\mcA^i,\exists B\in\mcA^{i+1}, A\subset S\subseteq B}$. Note that 
\begin{enumerate}
    \item $\mcA^{i+1}\setminus \mcA^i\subseteq \mcC^i$,
    \item $\mcA^i\cap\mcC^i=\phi$,
    \item any edge whose cost is reduced by 1/2 in iteration $i$ goes across a set in $\mcC^i$,
    \item $\mcC^i\cap\mcC^{i+1}=\phi$ for $i\in[T-1]$
\end{enumerate} 

Before $A\in\mcC^i$ was considered in iteration $i$ we would have considered the sets in $\mcS^i$ corresponding to children of node $v_A$ in tree $\mcT$. Let $B_j, j\in[p_i]$ be these sets and note that they belong to $\mcC^i\cup\mcA^i$. For each $B_j, j\in[p_i]$ we would already have reduced the cost of edges $e\in\delta(B_j)$ if $\pi_e(B_j)\neq\set{t^i}$.  Hence when considering $A$ we would only be reducing the cost of edges in $\delta(A)$ which are incident to $A\setminus\cup_{j\in[p_i]} B_j=X^i_A$. Thus the edges of $F$ whose cost is reduced when considering $A\in\mcC^i$ are subsets of $\alpha^i(A)$, let this subset be $\beta^i(A)$.

 After iteration $i$, (reduced) cost of an edge $e$ is $c_e-\sum_{S:e \in \delta(S)}y_S$, where $y$ is the dual value after iteration $i$. Note that as the algorithm proceeds, (reduced) cost of edges decrease. To prove that the modified WGMV procedure gives a 2-approximate solution, we bound the total reduction in costs of edges in $F$ by twice the total increase in the value of dual variables. In iteration $i$, the total reduction in edge costs of $F$ due to increase of dual variables of sets $A^i$ equals $\gamma^i\sum_{S\in\mcA^i} \abs{\delta_F(S)}=\gamma^i\sum_{S\in\mcA^i} \abs{\alpha^i(S)}$, where $\gamma^i=t^i-t^{i-1}$ is the increase in the dual variable of a set in $\mcA^i$. The other reduction occurs when we reduce by 1/2 the costs of edges due to parity considerations. The total reduction in the cost of edges of $F$ due to this reason is at most $1/2\sum_{S\in\mcC^i}\abs{\beta^i(S)}$.  

To prove the approximation guarantee of WGMV, authors in \cite{WGMV} show that in every iteration the total reduction in cost of edges in $F$ is at most twice the total increase in dual values in that iteration. To prove the approximation guarantee of modified WGMV, we need to charge the reduction in edge costs across iterations. To do this, we introduce a procedure for marking and unmarking sets. All sets are unmarked before the first iteration of the algorithm. In the first iteration a set $A\in\mcS$ is not marked 
\begin{enumerate}
\item if $A$ is critical or,
\item if node $v_A$ exhausts all its tokens and $\alpha^1(A)=\phi$ 
\end{enumerate}
All other sets in $\mcS$ are marked in iteration 1. Let $M$ be the number of sets which are marked. 

In iteration $i$ we unmark a set $S\in\mcC^i$ if it is critical and $\beta^i(S)\neq\phi$. Let $M_i$ be the number of sets unmarked in iteration $i$. In Lemma~\ref{lem:marking} we argue that we unmark a set only if it has a mark on it.

\begin{lemma}\label{lem:modified}
In any iteration $i>1$, $$\gamma^i\sum_{S\in\mcA^i}\abs{\alpha^i(S)} + (1/2)\sum_{S\in\mcC^i}\abs{\beta^i(S)} - M_i/2 \le 2\gamma^i(\abs{A_i}-1)$$
\end{lemma}
\begin{proof}
Recall $\mcR^i$ is the collection of critical sets in iteration $i$.
\begin{equation}\label{eq:one}
\gamma^i\sum_{S\in\mcS^i} \abs{\alpha^i(S)} \ge  \gamma^i\sum_{S\in\mcA^i} \abs{\alpha^i(S)} + \gamma^i\sum_{S\in\mcC^i} \abs{\alpha^i(S)} + \gamma^i\sum_{S\in\mcS^i\setminus\mcA^i\cup\mcC^i} \abs{\alpha^i(S)}
\end{equation}
By Lemma~\ref{lem:one} we obtain
\begin{equation}\label{eq:two}
\gamma^i\sum_{S\in\mcS^i\setminus\mcA^i\cup\mcC^i} \abs{\alpha^i(S)} \ge  \gamma^i\abs{\mcR^i\setminus\mcC^i}
\end{equation}
and the unmarking procedure gives
\begin{equation}\label{eq:three}
\gamma^i\sum_{S\in\mcC^i} \abs{\alpha^i(S)}+M_i/2 \ge  (1/2)\sum_{S\in\mcC^i} \abs{\beta^i(S)} + \gamma^i\abs{\mcR^i\cap\mcC^i}
\end{equation}
Inequality \ref{eq:three} holds since 
\begin{enumerate}
\item if $S$ is not critical it contributes $\gamma^i\abs{\alpha^i(S)}$ to the left and $\abs{\beta^i(S)}$ to the right and $\beta^i(S)\subseteq\alpha^i(S)$.
\item if $S$ is critical but $\beta^i(S)=\phi$ then $S$ contributes $\gamma^i\abs{\alpha^i(S)}$ to the left and $\gamma^i$ to the right and $\alpha^i(S)\neq\phi$.
\item if $S$ is critical and $\beta^i(S)\neq\phi$ then $S$ contributes $\gamma^i\abs{\alpha^i(S)}+1/2$ to the left and $(1/2)\abs{\beta^i(S)}$ $+ \gamma^i$ to the right. Since $\phi\neq\beta^i(S)\subseteq\alpha^i(S)$ and $\gamma^i\ge 1/2$, the contribution to the left is more than the contribution of $S$ to the right.
\end{enumerate}
Adding inequalities \ref{eq:one}, \ref{eq:two} and \ref{eq:three} we get
\begin{equation}\label{eq:four}
\gamma^i\sum_{S\in\mcS^i} \abs{\alpha^i(S)} \ge  \gamma^i\sum_{S\in\mcA^i} \abs{\alpha^i(S)} + \gamma^i\sum_{S\in\mcC^i} \abs{\beta^i(S)} + \gamma^i\abs{\mcR^i} -M_i/2
\end{equation}
Inequality \ref{eq:four} when combined with the inequality in Lemma~\ref{lem:equal} and together with the fact that $\gamma^i\ge 1/2$ proves the lemma.
\end{proof}

Iteration 1 differs from other other iterations since we mark sets in this iteration. For iteration 1 we make the following claim.
\begin{lemma}\label{lem:modified1}
$$\gamma^1\sum_{S\in\mcA^1}\abs{\alpha^1(S)} + (1/2)\sum_{S\in\mcC^1}\abs{\alpha^1(S)} + (1/2)(M-M_1)\le 2\gamma^1(\abs{A_1}-1)$$
\end{lemma}
\begin{proof}
Inequalities \ref{eq:one} and \ref{eq:three} remain unchanged for iteration 1 (with 1 replacing $i$) while inequality \ref{eq:two} is modified due to the marks placed on sets.
$A$ is marked if it is not critical and $\alpha^1(A)\neq\phi$; let $m$ be the number of such sets. This together with Lemma~\ref{lem:one} gives
\begin{equation}\label{eq:modtwo}
\gamma^1\sum_{S\in\mcS^1\setminus\mcA^1\cup\mcC^1} \abs{\alpha^1(S)} \ge  \gamma^1\abs{\mcR^1\setminus\mcC^1}+m/2
\end{equation}
Adding inequalities~\ref{eq:one}, \ref{eq:three} (with $i=1$) and inequality \ref{eq:modtwo} we get
\begin{equation}\label{eq:modfour}
\gamma^1\sum_{S\in\mcS^1} \abs{\alpha^1(S)} \ge  \gamma^1\sum_{S\in\mcA^1} \abs{\alpha^1(S)} + \gamma^1\sum_{S\in\mcC^1} \abs{\beta^1(S)} + \gamma^1\abs{\mcR^1} +(1/2)(m-M_1)
\end{equation}
Recall that we also mark a set $A$ when node $v_A$ does not exhaust all its tokens. Note that the number of such sets is $M-m$ and hence the inequality on Lemma~\ref{lem:equal} becomes
 \begin{equation}\label{eq:equal}
\sum_{S\in\mcS^1} \abs{\alpha^1(S)} + M-m \le 2(\abs{\mcA^1}-1)+\abs{\mcR^1}
\end{equation}
Combining inequality \ref{eq:modfour} and inequality~\ref{eq:equal} and using the fact that $\gamma^i\ge 1/2$ proves the lemma.
\end{proof}

\remove{
\begin{enumerate}
\item Once a set is marked in an iteration it remains marked until it is unmarked in a subsequent iteration. A set once unmarked is not marked again.
\item A set $S\in\mcS^i\setminus\mcA^i\cup\mcC^i$ is {\em marked} in iteration $i$ if $S$ is not critical in iteration $i$ and $\alpha^i(S)\ge 1$.
\item In iteration $i$ we unmark a set $S\in\mcC^i$ if it is critical and $\beta^i(S)\ge 1$. In Lemma~\ref{lem:marking} we argue that we unmark a set only if it has a mark on it.
\end{enumerate}}

Summing the inequality in Lemma~\ref{lem:modified1} and  Lemma~\ref{lem:modified}  over all iterations gives us 
\begin{equation*}
\sum_{i\in[T]}\left(\gamma^i\sum_{S\in\mcA^i}\alpha^i(S) + (1/2)\sum_{s\in\mcC^i}\beta^i(S) - M_i/2\right) + M/2 \le \sum_{i\in[T]} 2\gamma^i(\abs{A_i}-1)
\end{equation*}
Since we unmark a set only if it has been marked in iteration 1 (Lemma \ref{lem:marking}), $M\ge\sum_{i\in[t]} M_i$. Therefore, the total reduction in the cost of the edges of $F$ over all iterations (which is the total cost of edges in $F$) is at most the quantity on the left of the above inequality. Hence the cost of the solution $F$ is at most twice the total dual raised over all iterations and this completes the proof of Theorem \ref{thm : half integer WGMV}.

It remains to show that a set is unmarked only if it has been marked in iteration 1.
 \begin{lemma}\label{lem:marking}
 If $A\in\mcC^i$ is critical in iteration $i$ but not marked in iteration 1 then $\beta^i(A)=\phi$.
 \end{lemma}
 \begin{proof}
 Let $\set{B_j|j\in[p]}$ be the sets corresponding to children of $v_A$ and $X^1_A=A\setminus\cup_{j\in[p]} B_j$. If $H^1_A$ has a tree spanning nodes corresponding to sets $X^1_A$, $B_j, j\in[p]$, then edges of $\delta(A)$ would have equal parity. If $A$ becomes a minimal unsatisfied set at time $t$ then $B_1$ was active till time $t$. Therefore the parity of edges in $\delta(B_1)$ and hence those of all edges in $\delta(A)$ would equal \set{t} which would imply $\beta^i(A)=\phi$.
 
 Since $A$ is unmarked either it is critical in iteration 1 or $v_A$ exhausts all its tokens and $\alpha^1(A)=\phi$. In the former case we have a tree spanning nodes corresponding to sets $X^1_A$, $B_j, j\in[p]$. In the latter case if there is no such tree there would be a tree spanning nodes corresponding to sets $B_j, j\in[p]$ and no edge in $\delta_F(A)$ incident to $X^1_A$. Again, this implies that all edges in $\delta_F(A)$ have equal parity.
 \end{proof}

\section{Computing half-integral flow in Seymour graphs}
\label{sec:Seymour}
\def\2ECAP{\mbox{\tt 2ECAP}}

In this section, we describe the connection between multicommodity flows/multicuts and connectivity augmentation problems from \cite{KGS}. In particular, we will be interested in \2ECAP, a special case of the \UCC\ problem defined in \cite{KGS}. 
\begin{definition}
\textbf{2-edge connectivity Augmentation Problem (\2ECAP)}: Given an undirected graph (without loops but possible parallel edges) $G=(V,E \cup Y)$ and edge weights $w:E\rightarrow \mathbb{Z}_{\geq 0}$ find a minimum weight set of edges $E' \subseteq E$ such that each connected component of $(V,E' \cup Y)$ is 2-edge connected. 
\end{definition}
For every $S \subseteq V$, let $f: S \rightarrow \{0,1\}$ be defined as follows: $f(S)=1$ iff exactly one edge of $Y$ crosses the cut $(S,V \setminus S)$, otherwise it is zero. $\2ECAP$ can be formulated equivalently as: find a minimum weight subset of edges $E' \subseteq E$ such that at least $f(S)$ edges of $E'$ cross the cut $(S, V \setminus S)$. It is well known that $f$ as defined above is uncrossable and hence the WGMV algorithm can be used to compute a 2-approximate solution.

Now, we define the multicommodity flow problem. Let $G=(V,E)$ be a simple undirected graph with edge capacities $c:E\rightarrow\mathbb{Z}_{\geq 0}$ (called the supply graph) and $H=(V,F)$ be a simple graph each edge of which corresponds to a commodity and the endpoints of that edge are the source/sink of that commodity (called the demand graph). Given any $G$ and $H$, an instance of sum multicommodity flow asks for a feasible flow of maximum total value between the end points of demand edges. A minimum multicut is a set of edges of $G$ of minimum total weight whose removal disconnects endpoints of all the demand edges in $G$. It is easy to see that the value of minimum multicut (C) is always greater than the value of the maximum flow (F). Given a class of instances, the maximum value of the ratio between $C$ and $F$ is known as the $\texttt{flow-multicut gap}$ for the class. This gap is $\theta(\log k)$ for general $G,H$ while it is $O(1)$ for planar $G$ and arbitrary $H$. There is rich literature on proving $\texttt{flow-multicut gaps}$ \cite{garg1996approximate,garg1997primal,klein1993excluded}. 

If we restrict the flow to be integral (resp. half integral), we call the flow-multicut gap as the integral (resp. half integral) $\texttt{flow-multicut gap}$. An instance of the multicommodity flow/multicut problem is called a Seymour instance if the union of the supply and demand graphs is planar. In \cite{KGS}, the authors establish a \texttt{flow-multicut gap} of at most 2 for Seymour instances by showing that the problem of computing a multicut in a Seymour instance is equivalent to solving an appropriate instance of $\2ECAP$ in the planar dual of the supply and demand graph. Given a planar graph $G$, let $G^{*}$ denotes its planar dual. Formally,
\begin{lemma}[\cite{KGS}]
$C$ is a multicut for the instance $(G,H)$ if and only if $C^{*}$ is a feasible solution to $\2ECAP$ for the instance $(V^{*},E^{*}\cup F^{*})$.
\end{lemma}
The WGMV algorithm immediately gives a 2-approximation algorithm for multicuts in Seymour instances. In order to prove the \texttt{flow-multicut gap}, \cite{KGS} shows that the duals constructed by the WGMV algorithm correspond to flow paths in $G$ and that this correspondence is value preserving, ie. total flow is equal to the total value of the dual and if the duals constructed are integral (resp. half integral), then the corresponding flows are integral (resp. half integral). Formally, 
\begin{lemma}[\cite{KGS}]
There exists a flow of value $\sum_{S \subseteq V^{*}}y_S$ in $G$.
\end{lemma}
\cite{KGS} show how to extract a half-integral flow of value at least half of any given fractional flow and an integral flow of value at least half any given half integral flow. This shows a half integral (resp. integral) flow-multicut gap of 4 (resp. 8). Using our modified WGMV algorithm, we obtain a half integral dual (and hence half integral flow) of value at least half the cost of the $\2ECAP$ solution and hence the multicut. This gives us a 2 (resp. 4) approximate half-integral (resp. integral) \texttt{flow-multicut} theorem for Seymour instances. 

\begin{theorem} \label{integral maxflow mincut}
Let $G+H$ be planar. There exists a feasible integral flow of value $F$ and a multicut of value $C$ such that $C \leq 4F$. Further, such a flow and cut can be computed in polynomial time.
\end{theorem}
$\cite{KGS}$ shows a class of Seymour instances such that the half-integral flow-multicut gap approaches 2 from below. This, along with our upper bound of 2 proves that Theorem~\ref{half integral maxflow mincut} is tight. The best known lower bound for the integral \texttt{flow-multicut gap} is also 2 and it remains an interesting open question to determine the exact gap.

%


\bibliographystyle{plain}
\bibliography{WGMV.bib}

\end{document}